\documentclass[letterpaper, 10 pt, conference]{ieeeconf}  % Comment this line out if you need a4paper

\IEEEoverridecommandlockouts                              % This command is only needed if 
\usepackage{cite}

\usepackage{booktabs}
\usepackage{amsmath,amssymb,amsthm,amsfonts,xcolor}
\usepackage{lipsum}
\usepackage{graphicx}
\usepackage{float}
\usepackage{subfig}
\usepackage[export]{adjustbox}
\usepackage{tikz}
\usetikzlibrary {automata, positioning, arrows.meta,calc,cd}
\usetikzlibrary{shapes.geometric, arrows}
\usepackage{multicol}
\usepackage{multirow}
\usepackage{booktabs}
\newtheorem{definition}{Definition}
\newtheorem{example}{Example}
\newtheorem{theorem}{Theorem}
\newtheorem{remark}{Remark}
\newtheorem{problem}{Problem}

\renewenvironment{proof}{\textit{Proof:}}{\qedsymbol}
\renewcommand{\qedsymbol}{\hfill$\blacksquare$}

\usepackage[moderate,tracking=normal]{savetrees}

\title{\LARGE \bf
Robust Recovery and Control of Cyber-physical Discrete Event Systems\\ under Actuator Attacks*}

\author{Samuel Oliveira$^{1,2}$, Mostafa Tavakkoli Anbarani$^{3}$, Gregory Beal$^{3}$,  Ilya Kovalenko$^{3}$, Marcelo Teixeira$^{4}$,\\ André B. Leal$^{2}$ and Rômulo Meira-Góes$^{5}$% <-this % stops a space
\thanks{This work was supported in part by the Coordenação de Aperfeiçoamento de Pessoal de Nível Superior - Brasil (CAPES) – Finance Code 001 and grant CAPES/AUXPE 2910/2023 (Processo nº 88881.919326/2023-01-PDPG-CONSOLIDACAO-3-4) and by the Fundação de Amparo à Pesquisa e Inovação do Estado de SC, Brazil - FAPESC (2023TR000924).}% <-this % stops a space
\thanks{$^{1}$Department of Exact and Technological Sciences, Federal University of Amapá (UNIFAP), Macapá, AP, Brazil. (e-mail: samuel.oliveira@unifap.br)}%
\thanks{$^{2}$Graduate Program in Electrical Engineering, Santa Catarina State University (UDESC), Joinville, SC, Brazil. (e-mail: andre.leal@udesc.br)}%
\thanks{$^{3}$Department of Mechanical Engineering, The Pennsylvania State University (PSU), University Park, PA, USA. (e-mails: mkt5457@psu.edu, glb5367@psu.edu, iqk5135@psu.edu)}
\thanks{$^{4}$Graduate Program in Electrical Engineering, Federal University of Technology (UTFPR), Pato Branco, PR, Brazil. (e-mail: mtex@utfpr.edu.br)}
\thanks{$^{5}$School of Electrical Engineering and Computer Science, The Pennsylvania State University (PSU), University Park, PA, USA. (e-mail: romulo@psu.edu)}%
}

\begin{document}

\maketitle
\begin{tikzpicture}[remember picture,overlay]
\node[anchor=south,yshift=8pt] at (current page.south) {
    \parbox{\textwidth}{
        \footnotesize
        \textbf{IEEE Copyright Notice}\par
        \vspace{0.3em}
        \textcopyright~2025 IEEE. Personal use of this material is permitted.
        Permission from IEEE must be obtained for all other uses, in any current or
        future media, including reprinting/republishing this material for advertising
        or promotional purposes, creating new collective works, for resale or redistribution
        to servers or lists, or reuse of any copyrighted component of this work in other works.
    }
};
\end{tikzpicture}

\thispagestyle{empty}
\pagestyle{empty}

%%%%%%%%%%%%%%%%%%%%%%%%%%%%%%%%%%%%%%%%%%%%%%%%%%%%%%%%%%%%%%%%%%%%%%%%%%%%%%%%
\begin{abstract}

Critical real-world applications strongly rely on Cyber-physical systems (CPS), but their dependence on communication networks introduces significant security risks, as attackers can exploit vulnerabilities to compromise their integrity and availability.
This work explores the topic of cybersecurity in the context of CPS modeled as discrete event systems (DES), focusing on recovery strategies following the detection of cyberattacks. 
Specifically, we address actuator enablement attacks and propose a method that preserves the system's full valid behavior under normal conditions. Upon detecting an attack, our proposed solution aims to guide the system toward a restricted yet robust behavior, ensuring operational continuity and resilience. Additionally, we introduce a property termed AE-robust recoverability, which characterizes the necessary and sufficient conditions for recovering a system from attacks while preventing further vulnerabilities. Finally, we showcase the proposed solution through a case study based on a manufacturing system.
\end{abstract}

%%%%%%%%%%%%%%%%%%%%%%%%%%%%%%%%%%%%%%%%%%%%%%%%%%%%%%%%%%%%%%%%%%%%%%%%%%%%%%%%
\section{Introduction}
% \rmg{let's use attacker instead of intruder. Intruder is more used in opacity, it has a passive connotation.}
% \rmg{If we need space, we could make the intro shorter.}

Communication networks are an essential part of cyber-physical systems (CPS), enabling the exchange of information between physical components and their controller. 
As a result, these systems are vulnerable to the disclosure of private and confidential information and to malicious attacks that can disrupt normal operations. 
Such attacks can cause the system to deviate from its normal or intended behavior, potentially causing severe damage to its infrastructure and equipment, leading to high operational costs. 
Therefore, information confidentiality, service availability, and system integrity are important requirements of CPS. 

This work addresses these cybersecurity challenges within the framework of discrete event systems (DES) \cite{Lafortune:2021}. 
In DES, the system (\emph{plant}) is controlled by a supervisor that observes events generated by the plant to then disable events to satisfy a control specification. 
In the cybersecurity context, an attacker might compromise the network communication between the plant and the supervisor.
In this work, we focus on developing a method to ensure \emph{recoverability} and operational continuity in cyber-physical discrete event systems (CPDES) as a mitigation strategy against attacks.

A relevant number of papers focused on addressing cybersecurity in DES, following specific strategies, such as \textit{attack modeling} \cite{meira-goes:2020synthesis, Lin:2020}, \textit{intrusion detection and mitigation} \cite{Carvalho:2018, Wang:2020, Lima:2021, Fahim:2024,fahim2025enhancing, OLIVEIRA2024198}, \textit{synthesis of resilient supervisors} \cite{Su:2018, Ma:2022, meira-goes:2023dealing}, and \textit{system recovery} \cite{Pena:2022, cavalcanti:2025, Sakata:2023, Sakata:2024}. 
For a complete review of cybersecurity in the context of DES, we refer the reader to \cite{hadjicostis2022cybersecurity, OLIVEIRA2023100907}. 

While papers addressing attack modeling take the perspective of the attacker to formulate successful attacks, works on intrusion detection and mitigation focus on identifying such attacks. In addition, intrusion detection and mitigation methods aim to prevent the system from reaching unsafe states by disabling certain events \cite{Carvalho:2018, Lima:2021}. 
Although the integrity of the system is maintained by following this strategy, its continuity can be affected. 
Methods for synthesizing resilient supervisors, on the other hand, usually restrict the closed-loop behavior of the system to ensure safety \cite{Ma:2022}. 
These restrictions are enforced even in the absence of attacks.

System recovery methods aim to restore the operation of the system after an attack has occurred and been detected.
The work in \cite{Pena:2022} presents a secure recovery procedure based on synchronizing automata.
This procedure resynchronizes the supervisor and plant estimates after desynchronization due to attacks by forcing the system to return to its initial state.
While this approach provides recovery and operational continuity after attacks, the attacker's capabilities remain unchanged, i.e., the system continues to be vulnerable to new attacks.
In contrast, the work in \cite{cavalcanti:2025} explores the Bipartite Transition System (BTS) to find control strategies that recover the system from attack detection back to nominal behavior while avoiding unsafe states.
The authors assume the attacker is isolated immediately upon detection, meaning no further attacks can occur.
Our approach differs by not assuming immediate attacker isolation and focusing on a recovery method that operates under continued attacker presence.

Research on \emph{fallback control systems} explores a related idea. 
In \cite{Sakata:2023, Sakata:2024}, continued operation after attacks is achieved by switching to a fallback controller that is connected to the network but is assumed to be immune to attacks. 
This approach preserves nominal behavior and maintains availability, but relies on the existence of such an attack-immune fallback controller. 
In contrast, our approach does not require this assumption. 
Instead, we address the problem within the standard supervisory control theory framework, making our approach more general and applicable. %In this scenario, instead of addressing unsafe states, the main goal is to avoid blocking situations due to attacks.

The method proposed in this paper differs from the aforementioned works in that 
(i) in the absence of attacks, it allows the system to execute nominal behavior, (ii) upon attack detection, a recovery strategy steers the system to a subset of secure states, and (iii) a post-attack robust control keeps the system within secure states. 
Figure~\ref{fig:illustration} illustrates these three concepts.
After nominal behavior $s$, the attacker allows event $\sigma$ to disrupt the nominal operation.
Upon attack detection, our \emph{robust-recovery strategy} achieves three objectives: (i) prevents the system from reaching unsafe states, (ii) recovers the system from this undesired behavior to normal and secure operation (sequence of events $r$), and (iii) enforces a robust controller against attacks (shaded zone within controlled behavior). 
Our approach ensures operational continuity while preserving both integrity and availability, which are critical aspects of cybersecurity.

\begin{figure}[t!]
\vspace{.25cm}
    \centering
    \includegraphics[width=0.9\linewidth, page=2]{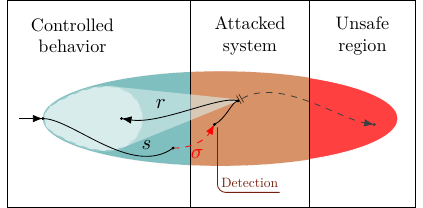}
    \caption{Robust-recovery strategy}
     \vspace{-10pt}
    \label{fig:illustration}
\end{figure}

In this paper, we focus on a specific type of attack known as \emph{actuator enablement attacks} (AE-attacks for short). 
In such attacks, some controllable events disabled by the supervisor can be re-enabled by an attacker, leading to unexpected execution of undesired (and possibly damaging) behaviors in the plant. 
Upon attack detection, we propose the synthesis of robust-recovery strategies that lead the system to a new post-attack behavior, ensuring it remains within a subset of states where AE-attacks are no longer feasible. 

The remaining sections of this work are organized as follows: Section~\ref{sec:pre} presents the basic concepts of automata and languages applied to supervisory control of DES under AE-attacks. Section~\ref{sec:problem} presents the problem of synthesizing robust-recovery strategies for CPDES. Section~\ref{sec:solution} presents our proposed approach. 
In Section~\ref{sec:case study}, the proposed approach is applied to a case study based on a manufacturing system. 
Finally, in Section~\ref{sec:conclusions}, the conclusions and future works are presented.

% \rmg{after the changes above, I think the intro looks good!}

\section{Preliminaries}\label{sec:pre}
\subsection{Supervisory control of DES}
We consider a DES modeled as a deterministic finite-state automaton (DFA) $G = (X, \Sigma, f, \Gamma, x_0, X_m)$, where $X$ denotes the finite set of states; $\Sigma$ is the set of events; $f : X \times \Sigma \rightarrow X$ denotes the possibly partial transition function; $\Gamma : X \rightarrow 2^{\Sigma}$ is the active event function, as $\Gamma(x)$ represents all events $\sigma$ for which $f(x,\sigma)$ is defined; $x_0$ is the initial state, and $X_m$ is the set of marked states.
Given $G_1$ and $G_2$ as two deterministic automata, the synchronous composition (which combines two or more automata into a single one) between $G_1$ and $G_2$ is denoted as $G_1 \| G_2$ \cite{Lafortune:2021}.

A finite-length sequence of events in $\Sigma$ is named a string, and the set of all strings formed by events in $\Sigma$ is denoted by $\Sigma^*$, including the empty string $\varepsilon$ (string with no events). Any subset of $\Sigma^*$ is called a language over $\Sigma$.  For a language $\mathcal{L}$, the notation $\overline{\mathcal{L}}$, the so-called prefix closure of $\mathcal{L}$, represents the set of all the preﬁxes of all the strings in $\mathcal{L}$. 
Notation $\mathcal{L}/s$ represent the post-language of $\mathcal{L}$ after the sequence $s$, i.e., $\mathcal{L}/s := \{t \in \Sigma^* : st \in \mathcal{L}\}$.

Automaton $G$ models two languages: the language generated by $G$ and the language marked by $G$. 
The language generated by $G$ is defined as $\mathcal{L}(G) := \{s \in \Sigma^*  |  f(x_0,s)$ is defined$\}$. The language marked by $G$ is defined as $\mathcal{L}_m(G) := \{s \in \mathcal{L}(G)  |  f(x_0,s) \in X_m \}$. In other words, $\mathcal{L}_m(G)$ is the subset of $\mathcal{L}(G)$ that contains all strings whose transition functions lead to a final state, representing the behavior of $G$ in which tasks are completed. 

% \review{$G_2$ is said to be a subautomaton of $G_1$, denoted by $G_2 \sqsubseteq G_1$, if $\forall s\in \mathcal{L}(G_2) : f_{G_2}(x_{0,G_2},s) = f_{G_1}(x_{0,G_1},s)$. Additionally, we say that $G_2$ is a stric-subautomaton of $G_1$, denoted by $G_2 \sqsubset G_1$, if, in addition to $G_2 \sqsubseteq G_1$, if $s \in \mathcal{L}(G_1)$ and $s \notin \mathcal{L}(G_2)$, there exists and $s' \in \overline{s} : f_{G_1}(x_{0,G_1}, s') \not\subseteq X_{G_2}.$ \cite{Cho:1989}}

Given two automata $G_1$ and $G_2$, $G_2$ is a \textit{strict subautomaton} of $G_1$, denoted $G_2\sqsubset G_1$, if $G_2$ is a copy of $G_1$ minus a set of states $Q$, i.e., $X_{G_2} = X_{G_1}\setminus Q$.
Formally, $G_2\sqsubset G_1$ if (i) $G_2$ is a subgraph of $G_1$, and (ii)  for every $x,y\in X_{G_2}$ such that $f_{G_1}(x,s)=y$ for some $s\in \mathcal{L}(G_1)$ then $f_{G_2}(x,s)=y$.
As a consequence, $G_2\sqsubset G_1$ implies $\mathcal{L}(G_2)\subset \mathcal{L}(G_1)$, see, e.g., \cite{Cho:1989,Lafortune:2021} for more details.  
% Given two automata $G_1$ and $G_2$, $G_2$ is a \textit{subautomaton} of $G_1$, denoted by $G_2 \sqsubseteq G_1$, if $G_2$ is a subgraph of $G_1$ such that $\mathcal{L}(G_2) \subseteq \mathcal{L}(G_1)$. 
%and every sequence $s \in \mathcal{L}(G_2)$ reaches the same state in $G_1$. 
% Additionally, $G_2$ is \textit{strict subautomaton} of $G_1$, denoted by $G_2 \sqsubset G_1$, if $G_1$ is a copy of $G_2$ minus 
% Formally, if $G_2$ is a subautomaton of $G_1$ and for every $x,y\in X_{G_1}$ such that $f_{G_2}(x,s)=y$ for some $s\in \mathcal{L}(G_2)$ then $f_{G_1}(x,s)=y$.
% if there exists a sequence $s$ that $G_1$ can generate but $G_2$ cannot, i.e., $\mathcal{L}(G_2) \subset \mathcal{L}(G_1)$ \cite{Lafortune:2021}.
%, meaning $G_2$ does not include all states from $G_1$, i.e., $X_{G_2} \subset X_{G_1}$ \cite{Lafortune:2021}.}

In supervisory control of DES~\cite{Ramadge:1987}, $\Sigma$ is partitioned as $\Sigma_c \mathbin{\dot{\cup}} \Sigma_{uc}$, denoting the sets of controllable and uncontrollable events, respectively. Given an uncontrolled plant $G$, a supervisor $S$ observes the events generated by the plant and disables events in $\Sigma_c$ to ensure a control specification. 
% \review{$S$ is then represented by a strict subautomaton of $G$ ($S \sqsubset G$), which implies that $\mathcal{L}(S) \subset \mathcal{L}(G)$}.

The specification to be achieved in $G$ under control of $S$ is represented by a language $K \subseteq \mathcal{L}(G)$, which specifies the desired behavior for $G$. 
In this work, we assume that $K$ expresses a state specification ensuring that the plant avoids a subset of unsafe states $X_{US} \subseteq X$, where damage occurs. 
Specifically, $K$ is defined as $K = \{s \in \mathcal{L}(G) \mid f(x_0, s) \notin X_{US} \}$.
Based on this, a closed-loop system $S/G$ is obtained with language $\mathcal{L}(S/G) \subseteq \mathcal{L}(G)$, see e.g. \cite{Lafortune:2021,Wonham:2018}. 
The supervisor $S$ is said to be \textit{safe} if $\mathcal{L}(S/G) \subseteq \overline{K}$.

$K$ is said to be controllable with respect to $G$ if $\overline{K} \Sigma_{uc} \cap \mathcal{L}(G) \subseteq \overline{K}$, meaning it remains unchanged under uncontrollable events. 
If $K$ is not controllable, then a controllable and nonblocking sublanguage of $K$ must be computed. %, denoted by $SupC(K, \mathcal{L}(G))$. 
% Therefore, $\mathcal{L}_m(S/G) = SupC(K, \mathcal{L}(G))$.
Hereinafter, we use $S$ interchangeably as the supervisor and its automaton representation.
% Since $K$ is defined by a strict subautomaton of $G$, the representation of $S$ is a strict subautomaton of $G$, i.e., $S \sqsubset G$.
Since $K$ specifies behaviors of $G$ that avoid unsafe states, the supervisor $S$ that enforces this specification is represented by a strict subautomaton of $G$, i.e., $S \sqsubset G$.

\subsection{CPDES under AE-attacks}
% \rmg{I think there a few important information missing. We never say "a supervisor $S$ that ensures a spec. And we never say that we assume that $S$ and $G$ are AE-safe controllable.}

We consider a CPDES in which some of the system components may be connected to the supervisor through a vulnerable communication network. 
In this sense, we consider an attacker who has the ability to modify a subset of controllable events $\Sigma_{c,v} \subseteq \Sigma_c$, representing vulnerable actuators from the plant $G$. 
Specifically, we consider \emph{actuator enablement attacks}, where an attacker enables an event that was originally disabled by the supervisor to prevent the system from reaching an unsafe state $x \in X_{US}$. By overriding the supervisor's disablement command, the system is then allowed to reach unsafe states.  

% \rmg{I think here we need to ``define" a supervisor $S$. We need to say that a supervisor $S$ ensures that these states are not reached. For hereon, this supervisor $S$ is "fixed".}

In \cite{Carvalho:2018}, a method for detecting AE-attacks is presented, along with a property called Actuator Enablement Safe Controllability (AE-safe controllability for short). 
This property is satisfied if it is possible to detect any occurrence of attacks and then disable a non-vulnerable controllable event before the plant reaches an unsafe state, as shown in the following example.

\begin{example}\label{ex:gs-models}
    Consider the plant $G$ shown in Figure~\ref{fig:example_plant}, where $\Sigma_{uc} = \{\sigma_2, \sigma_4, \sigma_5, \sigma_7, \sigma_{10}\}$ and $\Sigma_c = \{\sigma_1, \sigma_3, \sigma_6, \sigma_9, \sigma_{11}\}$, the set of unsafe states $X_{US} = \{10, 11\}$, the set of vulnerable controllable events $\Sigma_{c,v} = \{\sigma_6\}$. 
    The supervisor $S$ then disables event $\sigma_6$ at states 6 and 4.
    %, and an admissible supervisor $S$ for $G$ (Figure~\ref{fig:example_sup}). 
    % \begin{figure}[H]
    % \subfloat[Plant $G$]{\label{fig:example_plant}\includegraphics[width=1\linewidth, page=5]{images/cdc_recovery_models.pdf}}
    % \hfill
    % \subfloat[Supervisor $S$]{\label{fig:example_sup}\includegraphics[width=.68\linewidth, page=6]{images/cdc_recovery_models.pdf}}
    % \par 
    % \caption{Plant and supervisor models (Example~\ref{ex:gs-models}).}
    % \label{fig:example_models}
    % \end{figure}

    \begin{figure}[t]
        \centering
        \hspace{.22cm}
        \includegraphics[width=1\linewidth, page=5]{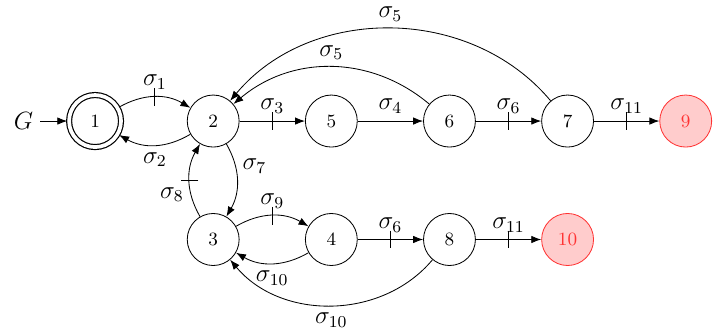}
        \caption{Plant $G$ for Example~\ref{ex:gs-models}.}
        \label{fig:example_plant}
    \end{figure}

    Even if an attacker enables event $\sigma_6$, the event $\sigma_{11} \in \Sigma_c \setminus \Sigma_{c,v}$ can be disabled to prevent the system from reaching unsafe states. 
    Thus, this system is said to be \emph{AE-safe controllable}.
\end{example}

Hereinafter, we assume: (i) the supervisor $S$ is \emph{safe}, i.e., it prevents the system from reaching unsafe states $x \in X_{US}$, and (ii) the closed-loop system $S/G$ is AE-safe controllable.

To model the system under attack, two new models are constructed: the attacked plant model $G^a$ and the attacked supervisor model $S^a$, both derived from the original models $G$ and $S$, respectively.
Let $\Sigma_{c,v}^a := \{\sigma^a : \sigma \in \Sigma_{c,v}\}$ be the set of vulnerable events that are improperly enabled by an attacker. 
Then $\sigma^a$ represents the occurrence of $\sigma$ in the plant when it is originally disabled by the supervisor, but enabled by an attacker. 

In the attacked plant model $G^a$, we add parallel transitions labeled with $\sigma^a$ for each original transition labeled with $\sigma \in \Sigma_{c,v}$, expressing the attacker’s ability to enable these events.
In the attacked supervisor model $S^a$, we add a new state labeled $A$ with a self-loop with all events $\Sigma \cup \Sigma_{c,v}^a$. 
Additionally, we add transitions in every state for each event $\sigma^a \in \Sigma^a_{c,v}$, except the ones where the original event $\sigma$ is feasible. 
These new transitions lead to the new state $A$.

State $A$ models the supervisor's actions after the attack.
For now, we assume that the supervisor allows every event.
Later on, we pose the robust-recovery problem to ensure a safe recovery, i.e., restrict the behavior post-attack.
The model of the closed-loop system under attack $S^a / G^a$ is obtained by computing $S^a \| G^a$.
Since every event is allowed after an attack, we obtain the complete behavior of $G$ after AE-attacks.

%Similarly, we add self-loops for each event $u \in \Sigma_{uc}$ in all states where the original event is not available in its set of active events to model the possible occurrences of uncontrollable events after AE-attacks. 

\begin{example} \label{ex:attacked-system}
Consider the plant $G$ and supervisor $S$ from Example~\ref{ex:gs-models}, with vulnerable events $\Sigma_{c,v} = \{\sigma_6\}$. The attacked supervisor $S^a$ and closed-loop system under attack $S^a/G^a$ are shown in Figure~\ref{fig:attacked_models}. 
Due to space limitation, $G^a$ is omitted.
%     \begin{figure}[t]
%     % \centering
%     \includegraphics[width=.9\linewidth, page=9]{images/cdc_recovery_models.pdf}
%     \caption{Attacked system $S^a/G^a$ for Example~\ref{ex:gs-models}.}
%     \label{fig:attacked-clsystem}
% \end{figure}

\begin{figure}[t]
    \hspace{.22cm}
    \subfloat[$S^a$]{
        \label{fig:attacked_sup}
        \centering
        \includegraphics[width=.75\linewidth, page=8]{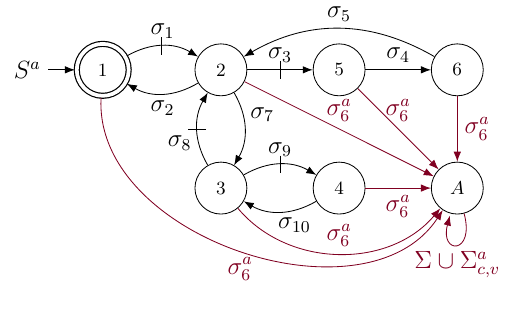}
    }
    \hfill
    \subfloat[$S^a/G^a$]{
        \label{fig:attacked-clsystem}
        \centering
        \includegraphics[width=1\linewidth, page=9]{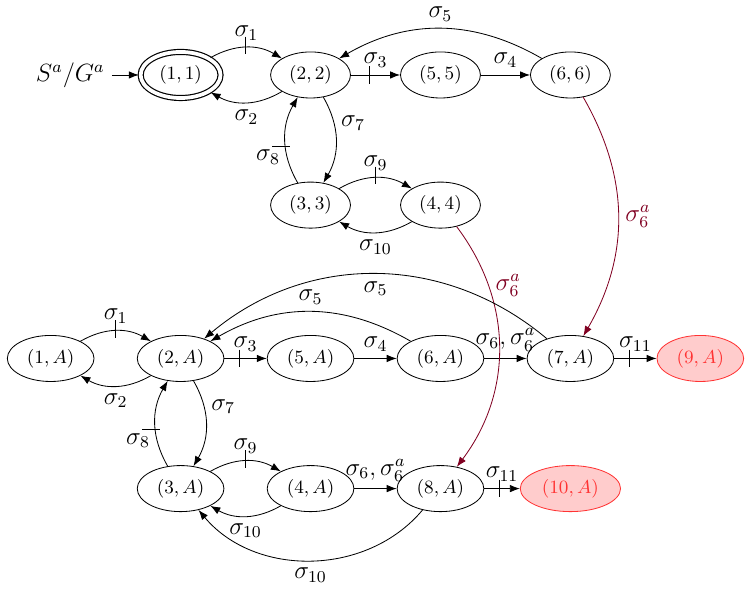}
    }
    \par 
    \caption{Attacked models (Example~\ref{ex:gs-models}).}
    \label{fig:attacked_models}
\end{figure}

% \samuel{Do we need to show Fig~\ref{fig:attacked-clsystem}?  Or is it enough to simply explain that after attacks, $S^a$ allows the entire behavior of $G$?}
% \rmg{Let see if we have space. If we don't, we can just say it.}

% Based on the mitigation strategy in \cite{Carvalho:2018}, after detecting the attack the event $\sigma_{11}$ is disabled to prevent the system from reaching unsafe states.
% \rmg{We say all controllable events are disabled, including $\sigma_{11}$}
\end{example}

As a mitigation strategy, \cite{Carvalho:2018} proposed disabling all controllable events of $G$ to prevent an unsafe behavior upon attack detection. 
Based on this, \cite{Wang:2020} extended the attack models to probabilistic discrete event systems (PDES) and proposed a mitigation strategy that minimizes the probability of unsafe behavior by enabling specific events after attacks, instead of disabling them all. 
Alternatively, \cite{Lima:2021} introduced a security supervisor that avoids unsafe states by disabling specific events only when there is a real risk of executing an unsafe behavior due to attacks. 

While these approaches maintain system integrity by preventing the system from reaching unsafe states, \emph{they cannot ensure operational continuity}.
Operational continuity is essential for maintaining the system's functionality after attacks.

\section{Robust-recovery of CPDES}\label{sec:problem}

In this section, we pose the problem of recovering a CPDES from AE-attacks and then enforcing a robust post-attack control strategy. 
In this scenario, the system is allowed to execute its full valid behavior in the absence of attacks. 

Once an attack is detected, the system should not only avoid unsafe states but also be able to transition to and remain in a robust region, which is composed of a given subset of states of $G$ where AE-attacks are ineffective.
This problem raises a fundamental question: \emph{Is it possible to transition the system from attack detection to such a region while avoiding further attacks?}

To address this question, we formalize: (i) the robust region that represents the robust behavior to be enforced after attacks, and (ii) the problem of synthesizing robust-recovery strategies for achieving this transition safely.

\subsection{Post-attack robust behavior}

The robust behavior is a subset of $\mathcal{L}(G)$ that guarantees system functionality while avoiding states that are vulnerable to AE-attacks. 
We start by defining the notion of vulnerable states.
\begin{definition}
    (Vulnerable states). Given a plant $G$ and the set of vulnerable controllable events $\Sigma_{c,v}$, the set of vulnerable states is defined as \begin{equation}
        X_v := \{x \in X \mid \Gamma(x) \cap \Sigma_{c,v} \neq \emptyset \}.
    \end{equation}
\end{definition}

Intuitively, the set of vulnerable states $X_v$ consists of states $x \in X$ where a vulnerable controllable event $\sigma \in \Sigma_{c,v}$ is feasible, making it susceptible to AE-attacks.

\begin{definition}\label{def:robust-beh}
    (Robust Behavior). Given a plant $G$ and the set of vulnerable states $X_v$, the robust behavior of $G$ is represented by a strict subautomaton $G_{robust} \sqsubset G$ such that  
    \begin{equation}
        X_{G_{robust}}  \cap (X_v \cup X_{US}) = \emptyset.
    \end{equation}
\end{definition}

Since $G_{robust}$ contains no vulnerable states $x \in X_v$, AE-attacks are not feasible within it. 
From here on, we refer to $G_{robust}$ as the \emph{robust region} of $G$.

\begin{remark}
    In this paper, we assume that $G_{robust}$ is given. Therefore, we do not impose additional constraints on its definition.
\end{remark}

\begin{example}
Consider the plant model $G$ presented in Figure~\ref{fig:example_plant} with $\Sigma_{c,v} = \{\sigma_6\}$. 
A robust region $G_{robust}$ with respect to $G$ and $\Sigma_{c,v}$ is formed based on states $\{1,2,3\}$. 
% Figure~\ref{fig:k-robust}.
    % \begin{figure}[H]
    %     \centering
    %     \includegraphics[width=0.5\linewidth, page=11]{images/cdc_recovery_models.pdf}
    %     \caption{The robust region $G_{robust}$ for Example~\ref{ex:gs-models}.}
    %     \label{fig:k-robust}
    % \end{figure}
\end{example}

\begin{remark}
    Note that there is a trade-off between the number of vulnerable events in $G$ and the restrictiveness of $G_{robust}$. As the number of vulnerable events increases, the problem becomes more challenging, possibly requiring stricter constraints on the robust behavior.
\end{remark}
% \rmg{I think the definition of robust region is still confusing. We need to work on it. 
% First, let just define it as a subautomaton. 
% That is our definition. 
% The condition is that the states in this subautomaton are not vulnerable nor unsafe.
% Krobust is defined based on Grobust.
% }

\subsection{Robust-recovery from attacks} 
% Instead of restoring nominal controlled behavior after attacks, which would allow attackers to continually exploit its vulnerabilities, 
We formalize \emph{robust-recovery} in such a way that, after an attack, the system transitions to $G_{robust}$ while preventing further interference from attackers.
% \rmg{This is somewhat confusing based on what came before. First, what is nominal. We need to clearly define before (\mathcal{L}(S/G) is the nominal behavior). Unclear that attacker continually exploit. 
% Let's bring the example? Or should we have the example earlier.
% The options in Fig.2: 1) controller disables all controllable. System will get stuck in states 1, or 2, or 3. 
% 2) If the system returns to a nominal state, then nominal controller kicks in again. However, the attacker can still manipulate states.
% I am not sure how to convey this message neatly. But to me, it is still unclear the problems after reading the paragraph above.}

To do so, it is necessary to determine whether a robust-recovery strategy exists for all possible AE-attacks with vulnerable actuators $\Sigma_{c,v}$. 
Before formalizing robust-recovery strategies, we introduce some additional definitions related to attack detection that are required for this formulation.

% \subsubsection{Attack detection}
% Building upon the model of the attacked closed-loop system $S^a/G^a$ presented in the previous section, we now define the detection language of $G$.

\begin{definition}
    (Detection language). Given the attacked closed-loop system $S^a/G^a$, the detection language is defined as
    \begin{equation}\begin{aligned}
            \mathcal{L}_{det} := \{ s\sigma \in \mathcal{L}(G) \mid & \; s \in \mathcal{L}(S/G), \; \sigma \in \Sigma_{c,v}, \\
            & \; s\sigma^a \in \mathcal{L}(S^a/G^a) \}.
    \end{aligned}\end{equation}
\end{definition}  

In words, the detection language $\mathcal{L}_{det}$ consists of all valid strings $s$ that are followed by the execution of an attacked event $\sigma \in \Sigma_{c,v}$ while being recognized by the closed-loop system under attack $S^a/G^a$. 
Based on $\mathcal{L}_{det}$ and the assumption that $S\sqsubset G$, we now define the set of states where attack detection is achieved.
% \rmg{Why are we not defining detection states based on system $S^a/G^a$? We can only guarantee uniqueness if S/G is a subautomaton of G.}

\begin{definition}
    (Detection states). Given an attacked system $S^a/G^a$ and its detection language $\mathcal{L}_{det}$, the set of detection states is defined as
    \begin{equation}
        X_{det} := \{x \in X \mid \exists s\sigma \in \mathcal{L}_{det}, f(x_0, s\sigma) = x \}.
    \end{equation}
\end{definition}

The set of detection states is composed of states that are reached after the occurrence of AE-attacks. 
As we assume that these attacks can be observed by the attacked supervisor $S^a$, we also assume that once these states are reached, the attack is also detected. 
With the detection states defined, we now address the \emph{robust-recovery strategies}.

While there might exist multiple paths from a detection state to states in the robust region $G_{robust}$, some of these paths might reach vulnerable states. Such paths are undesirable when synthesizing recovery strategies, as they could expose the system to further attacks during the recovery process. To deal with such vulnerabilities, robust-recovery is formalized as follows.

\begin{definition}
    (Robust-recovery strategies). Given a detection state $x_d \in X_{det}$ and the robust region $G_{robust}$, a robust-recovery strategy is a method that transitions the system to $G_{robust}$ by a path $r$ from $x_d$ to any state $x \in X_{G_{robust}}$. This path must satisfy the following conditions:
    
    \begin{itemize}
        \setlength{\itemindent}{.3cm}
        \item[(1)] $f(x_d,r) \in X_{G_{robust}}\subset X_G$;

        \item[(2)] 
        $
        \!
        \begin{aligned}[t]
            \forall t \in \overline{r}, \; f(x_d, t) \notin X_v\cup X_{US};    % Safe
        \end{aligned}
        $ 
        
        \item[(3)] 
        $
        \!
        \begin{aligned}[t]
            \forall t \in \overline{r}, \; 
            &\nexists u \in \Sigma_{uc}^* \mid f(x_d,tu) \in X_v\cup X_{US};    %controllability
        \end{aligned}
        $ 

        \item[(4)]  $\forall t \in \overline{r}, \forall e \in \Sigma_{uc}$, if 
                    $f(x_d, te)! \; \land \; te \notin \overline{r}$ then from 
                    $f(x_d, te)$ a path $r'$ must satisfy conditions (1), (2), (3) and (4).
    \end{itemize}
\end{definition}

% Intuitively, this definition ensures that, from a detection state $x_d \in X_{det}$, there exists a path $r$ that (i) leads the system to a state in $G_{robust}$, (ii) remains within states neither vulnerable nor unsafe: $X \setminus (X_v\cup X_{US})$, avoiding any vulnerable and unsafe states, (iii) ensures that no state along $r$ contains an uncontrollable sequence leading to vulnerable or unsafe states, and (iv) guarantees that if an uncontrollable event occurs that diverts the system from path $r$, the resulting state must also have its own path $r'$ to $G_{robust}$ satisfying all these same conditions (1)-(4) recursively. 

Intuitively, this definition ensures we have a recovery $r$ from a detection state $x_d \in X_{det}$.
Condition (1) ensures that  $r$ leads the system to a state in $G_{robust}$.
Next, condition (2) enforces that $r$ does not visit vulnerable and unsafe: $X_v\cup X_{US}$, avoiding any vulnerable and unsafe states.
Condition (3) ensures that no state along $r$ contains an uncontrollable sequence leading to vulnerable or unsafe states.
Lastly, condition (4) guarantees recursively that if an uncontrollable event occurs that diverts the system from path $r$, the resulting state must also have a path $r'$ to $G_{robust}$ satisfying all these same conditions (1)-(4). 

% Figure \ref{fig:ae-safe-recoverability} provides an intuition behind the robust-recovery strategies.
% \begin{figure}[t]
%     \centering
%     \includegraphics[width=1\linewidth, page=4]{images/cdc_recovery_models.pdf}
%     \caption{An illustration of robust-recovery strategies from detection to $G_{robust}$.}
%     \label{fig:ae-safe-recoverability}
% \end{figure}

Based on the above, a detection state $x_d \in X_{det}$ is said to be \emph{AE-robust recoverable} if there exists a robust-recovery strategy leading from $x_d$ to any $x \in X_{G_{robust}}$. 
If all detection states are AE-robust recoverable, then the system $G$ is said to be AE-robust recoverable as defined in the following.

\begin{definition}\label{def:robust-recoverability}
    (AE-robust recoverable system) Given an attacked system $S^a/G^a$ and the set of detection states $X_{det}$, we say that the system $S^a/G^a$ is AE-robust recoverable if, for every detection state $x_d \in X_{det}$, there exists a robust-recovery strategy leading to $G_{robust}$.
\end{definition}

\begin{example}
    Consider the attacked system $S^a/G^a$ shown in Figure~\ref{fig:attacked-clsystem}. Observe that after the attacked event $\sigma_6^a \in \Sigma_{c,v}^a$, the system can always reach a state $x \in X_{G_{robust}}$ without reaching any vulnerable states. Specifically, this is possible from state 7 through the uncontrollable event $\sigma_5$ and from state 8 through the uncontrollable event $\sigma_{10}$. %Hence, the system $S^a/G^a$ is said to be \emph{AE-Robust recoverable}.
\end{example}

% Having formalized robust-recovery strategies, we can now pose the problem of synthesizing robust-recovery strategies (SRRS).
% \rmg{The way we have right not it is weird to pose the synthesis problem. So far, we only talked about a system being or not robust-recoverable (verification). And not as something to enforce (synthesis).
% It is more of a presentation problem than a technical problem.}
% \rmg{I think this problem would be better posed as a synthesis problem since it is related to strategies.
% Something like:
% "Given an attacked system $S^a/G^a$ and a robust behavior $K_{robust}$, synthesize for each detection state $X_{det}$ an AE-robust-recoverable strategy, if they exist."
% }

The concept of AE-robust recoverability inherently involves two fundamental problems: verification and synthesis. Verification involves determining whether a CPDES meets the criteria for AE-robust recoverability, while synthesis focuses on developing methods to synthesize robust-recovery strategies if they exist.

\begin{problem}[Verification of AE-robust recoverability]\label{problem:strat-existence}
    Given an attacked system $S^a/G^a$, the set of detection states $X_{det}$ and the robust region $G_{robust}$, determine whether $S^a/G^a$ is AE-robust recoverable.
    % i.e., verify whether for every detection state $x_d \in X_{det}$, there exists a robust-recovery strategy leading to $G_{robust}$.
\end{problem}

\begin{problem}[Synthesis of robust-recovery strategies]\label{problem:strat-synth}
    Given an AE-robust recoverable attacked system $S^a/G^a$ w.r.t.   robust region $G_{robust}$, synthesize a robust-recovery strategy for every detection state $x_d \in X_{det}$.
\end{problem}

To address these problems, we propose a solution that aims to synthesize a resilient supervisor that preserves nominal behavior under normal conditions and enforces robust-recovery strategies upon detecting attacks. 
We show that this approach solves both the verification and synthesis problems.

\section{Synthesizing Robust-Recovery Strategies}\label{sec:solution}

We formulate the problem of synthesizing robust-recovery strategies as a \emph{supervisory control problem}. 
In this sense, we construct a resilient control specification $H^a$ that captures both the pre- and post-attack behaviors, from which the resilient supervisor can be obtained.

% \subsection{Labeled model for the attacked closed-loop system}
% To differentiate the states of the closed-loop system before and after an attack, we use a labeler automaton $A_l$ that marks states with specific labels: normal operation (label $N$) and operation under attack (label $A$). The labeler automaton $A_l$ for the attacked system presented in Example~\ref{ex:attacked-system} is shown in Figure~\ref{fig:al}.

% \begin{figure}[H]
%     \centering
%     \includegraphics[width=0.35\linewidth, page=10]{images/cdc_recovery_models.pdf}
%     \caption{Labeler automaton}
%     \label{fig:al}
% \end{figure}

% The labeled attacked system, denoted by $S^a/G^a$, is obtained by computing $S^a/(G^a \| A_l)$. As a result, the states of $S^a/G^a_{\text{labeled}}$ exhibit the form $(x,N)$ or $(x,A)$, where $x$ is a state in $G^a$. 
% Specifically, a state $x$ is labeled with $A$ if the sequence leading to $x$ contains an event $\sigma^a \in \Sigma_{c,v}^a$; otherwise, it is labeled with $N$. 

% For convenience, we use the notation $x_N$ to refer to states before attack detection (normal states) and $(x,A)$ to refer to states after attack detection (post-attack states). 
% Based on $S^a/G^a_{\text{labeled}}$, we define a resilient control specification $H^a$.

\subsection{Resilient control specification $H^a$}

% \rmg{Take a look. I've changed some of the text because I thought we neede to be a bit clearer. I've introduced the system $G^R=S^a/G^a$. This is the system we want to control instead of $G^a$.}

The specification $H^a$ is modeled as a state-based specification, meaning the plant is restricted to avoid a subset of undesirable states. 
Specifically, $H^a$ is derived from the attacked system model $S^a/G^a$, where the plant follows the nominal controlled behavior before attacks. 
However, after an attack, the system's behavior expands to the entire behavior of $G$. 

First, we define $G^R = S^a/G^a$ as our system of interest, i.e., $G^R$ is the closed-loop behavior with and without attacks.
Figure~\ref{fig:attacked-clsystem} depicts $G^R$ for our running example.
At this point, we aim to restrict $G^R$ in two ways after attacked events: (i) to avoid reaching vulnerable and unsafe states, and (ii) to ensure the system can reach a state within the robust region $G_{robust}$. 

In the case of (i), we create the specification $H^a$ from $G^R$.
Specification $H^a$ is a copy of $G^R$, but we remove all vulnerable and unsafe states from the attacked state space $(x,A)$ for $x\in X_v\cup X_{US}$.
Regarding (ii), we redefine the marked states in $G^R$ to be states in $G_{robust}$ after attacks.
Formally, $X_{m}^R$ marks all states $(x,A)$ for $x \in X_{G_{robust}}$.
This new marking specifies that states in the robust region must be reachable after attacks.
% At this point, we aim to restrict $G^R$ in two ways: (i) to avoid reaching vulnerable and unsafe states and (ii) to ensure the system can reach a state within the robust region $G_{robust}$. 

% Next, we mark all states $(x,A)$ for $x \in X_{robust}$, i.e., states in $G_{robust}$.
% This is achieved by removing all  and also by marking  all states from $G_{robust}$ within the attacked state space. This represents our reach goal after attacks.
% Finally, we construct the resilient control specification as:  
% \begin{equation}
%     \begin{aligned}
% H^a = \{ s \in S^a/G^a \mid \; &\forall (x,A) \in X_{S^a/G^a} \cap X_v, \\ &f(x_0,s) \neq (x,A) \}.
%     \end{aligned}
% \end{equation}

% $H^a$ consists of all sequences from $S^a/G^a$ that avoid post-attack vulnerable states $(x,A)$ for $x \in X_v$. 
% % Within the attacked state space, this restriction aligns with $G_{robust}$, which prevents the system from reaching any vulnerable states, as stated in Def.~\ref{def:robust-beh}. 
% Thus, enforcing $H^a$ allows the system to operate normally before an attack, even reaching normal vulnerable states $(x,y)$ for $x\in X_v$ and $y\in X_S$. 
% However, after attack detection, the system can transition to $G_{robust}$ if there are robust-recovery strategies, preventing further attacks and enhancing resilience.

\subsection{Resilient Supervisor $S^R$}  
In this work, we use the term \emph{resilient supervisor} to refer to a supervisor that achieves the nominal behavior in the absence of attacks and allows the system to reach any state within $G_{robust}$ after any possible AE-attack. We formally define it as follows:  

\begin{definition}\label{def:res-sup}
    (Resilient supervisor). Given an attacked system $G^R$ and a robust region $G_{robust}$, a supervisor $S^R$ is said to be resilient w.r.t. $\Sigma_{c,v}$ and $G_{robust}$ if after any AE-attack, it enforces a robust-recovery strategy.
\end{definition}

The resilient supervisor $S^R$ can be obtained by solving the standard supervisory control problem (SSCP) under the Ramadge-Wonham paradigm \cite{Ramadge:1987}, formulated as follows:

\begin{problem}\label{problem:sup-synth}  
(SSCP). Given system $G^R$ and a resilient control specification $H^a$, find (if it exists) a non-blocking supervisor $S^R$ for $G^R$, such that
\begin{itemize}
    \setlength{\itemindent}{.3cm}
    \item [(1)] $S^R$ is safe, i.e., $\mathcal{L}(S^R/G^R) \subseteq \mathcal{L}(H^a)$;
    \item[(2)] $S^R$ is maximally-permissive.
    % i.e., for any non-blocking supervisor $S^{R'}$ such that $\mathcal{L}(S^{R'}) \subseteq \mathcal{L}(H^a)$, $\mathcal{L}(S^R/G^R) \not\subset \mathcal{L}(S^{R'}/G^R)$. 
\end{itemize}
\end{problem}  

Once a supervisor $S^R$ is obtained, it is necessary to determine whether it satisfies the resilience property.
Since $H^a$ only deletes states from $G^R$ ($H^a\sqsubset S^R$), we have that supervisor $S^R\sqsubset G^R$. %, leading to the following problem:  

% \rmg{We need to be careful here. If we go back to having the problem with $G^a$, the supervisor $S^R$ is not a strict subautomaton of $G^a$. If we use the $G^R$, then $S^r\sqsubset G^R$ but $X_{G^R}\subseteq X_G\times X_{S^a}$. The theorem makes sense but we need to have the correct subset. }
\begin{theorem}\label{theorem:det-res-sup}
Supervisor $S^R$ is resilient if and only if all detection states $X_{det}$ are contained in $X_{S^R}$.
\end{theorem}

\begin{proof} [Sketch]
    We prove Theorem~\ref{theorem:det-res-sup} by contraposition. Suppose $S^R$ is not resilient, which means that there exists at least one possibility of AE-attacks after which $S^R$ does not enforce a robust-recovery strategy. 
    Consider a path $r$ from a detection state $x_d \in X_{det}$ to any state within $X_{G_{robust}}$. 
    For $S^R$ to fail to enforce a robust-recovery strategy from $x_d$, this path $r$ must not be contained within $S^R$. 
    This can occur for one of three reasons: (1) $r$ reaches vulnerable or unsafe states that have been removed from $H^a$; (2) Along path $r$, there exist uncontrollable events that lead to vulnerable or unsafe states also removed from $H^a$; or (3) Along path $r$, uncontrollable events lead to states that either have no path to $X_{G_{robust}}$, or any such path exhibits the characteristics described in reasons (1) and (2). 
    In any of these cases, the paths that would allow reaching the robust region must be disabled to ensure safety and non-blockingness (as shown in Problem~\ref{problem:sup-synth}) and maintain robustness against attacks. 
    This leads to the trimming of the recovery path from this point backwards and eventually deleting the detection state $x_d$. 
    Therefore, if $S^R$ is not resilient, then not all detection states $X_{det}$ are contained in $X_{S^R}$.
    
    The other direction follows a similar reasoning.
    Assume all detection states are contained in $X_{S^R}$.
    By the definition of $S^R$, from every detection state, there exists a safe path that reaches one of the marked states in $G^R$.
    In other words, a robust-recovery strategy leading the system to $G_{robust}$.
\end{proof}

% \review{If $S^R$ satisfies the resilience property, then the system is AE-robust recoverable.}
We now show that the problem of determining AE-robust recoverability and synthesizing robust-recovery strategies, i.e., Problems~\ref{problem:strat-existence} and ~\ref{problem:strat-synth}, can be reduced to the standard supervisory control problem, i.e., Problem~\ref{problem:sup-synth}.

\begin{theorem}
There exists a resilient supervisor $S^R$ for $G^R$ if and only if $S^a/G^a$ is \emph{AE-robust recoverable} w.r.t. $G_{robust}$. 
\end{theorem}

\begin{proof}
We start with the \emph{only if} part ($\Rightarrow$). Assuming that $S^R$ is resilient w.r.t. $G^R$ and $G_{robust}$, then it can always transition to $G_{robust}$ after attack detection. Thus, from any detection state $x_d \in X_{det}$, a valid recovery path to $G_{robust}$ exists while avoiding unsafe and vulnerable states. This satisfies AE-robust recoverability (Def.~\ref{def:robust-recoverability}).

We now show the \emph{if} part ($\Leftarrow$). Assuming that $G$ is AE-robust recoverable, every detection state $x_d \in X_{det}$ has a valid recovery strategy to $G_{robust}$. Given this property, a supervisor $S^R$ can be obtained such that it allows the system to follow these paths upon attack detection and achieve $G_{robust}$ afterward, fulfilling the definition of resilient supervisor (Def.~\ref{def:res-sup}).
\end{proof}

Consequently, a resilient supervisor $S^R$ provides a solution to both the verification and synthesis problems (Problems~\ref{problem:strat-existence} and~\ref{problem:strat-synth}). 
That is, the existence of a resilient supervisor $S^R$ implies that $S^a/G^a$ is AE-robust recoverable.
Conversely, if $S^a/G^a$ is AE-robust recoverable, then a resilient supervisor $S^R$ can be synthesized to enforce all robust-recovery strategies from detection states to $G_{robust}$.

\begin{example}
    Consider $G$ and $S$ as given in Example~\ref{ex:gs-models}. Based on these models, we construct the attacked system $G^R$, and then the resilient control specification $H^a$. Finally, the resilient supervisor $S^R$ for plant $G^R$ is shown in Figure~\ref{fig:resilient-supervisor}.
    \begin{figure}[]
    \hspace{.15cm}
    \centering
    \includegraphics[width=.9\linewidth, page=13]{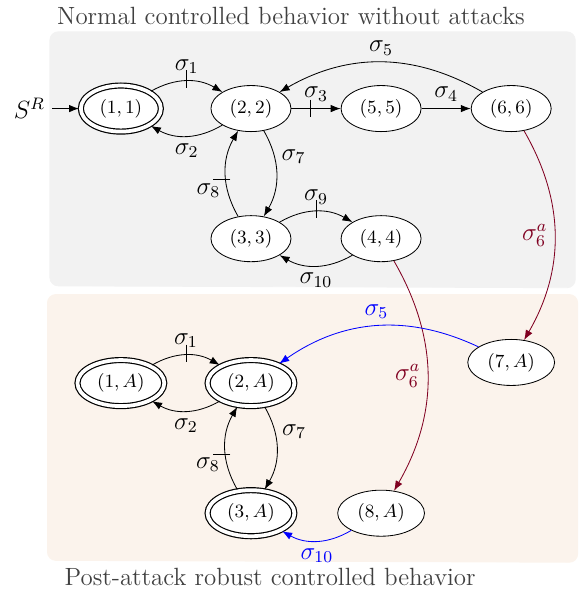}
    \caption{Resilient supervisor $S^R$.}
    \label{fig:resilient-supervisor}
\end{figure}
Note that $S^R$ contains the desired behavior for the system both before and after AE-attacks. In detail, the transitions highlighted in red refer to the attacked events while the transitions highlighted in blue refer to the recovery paths that lead to $G_{robust}$.
\end{example}

\section{Case Study: Manufacturing system}\label{sec:case study}
In this section, we demonstrate our robust-recovery strategy through a case study on a manufacturing testbed.

\subsection{Fischertechnik factory}
We use the Fischertechnik factory simulation 24v\textsuperscript{\texttrademark}, a down-scale manufacturing testbed shown in Fig.~\ref{fig:case-study-Fischertechnik}.  
The testbed consists of five stations: industrial oven, turntable, sorting station, conveyor belt, robotic gripper, and warehouse. 
In this case study, we modeled the behavior of the sorting station, conveyor belt, robotic gripper, and warehouse.

Two types of parts arrive on the conveyor belt: blue or red parts.
However, only one part is in the system at any given time.
Next, the conveyor moves the part to the sorting station.
The sorting station has three different locations: initial, location 1, and location 2.
In locations 1 and 2, pneumatic cylinders can push the parts out of the conveyor belt to buffers.
Lastly, the gripper picks parts for storage.
The gripper can pick parts from the initial location of the conveyor belt and the two buffers. 

\begin{figure*}[t]
    \subfloat[Fischertechnik manufacturing testbed]{
        \label{fig:case-study-Fischertechnik}
        \centering
        \includegraphics[width=.4\linewidth, valign=m]{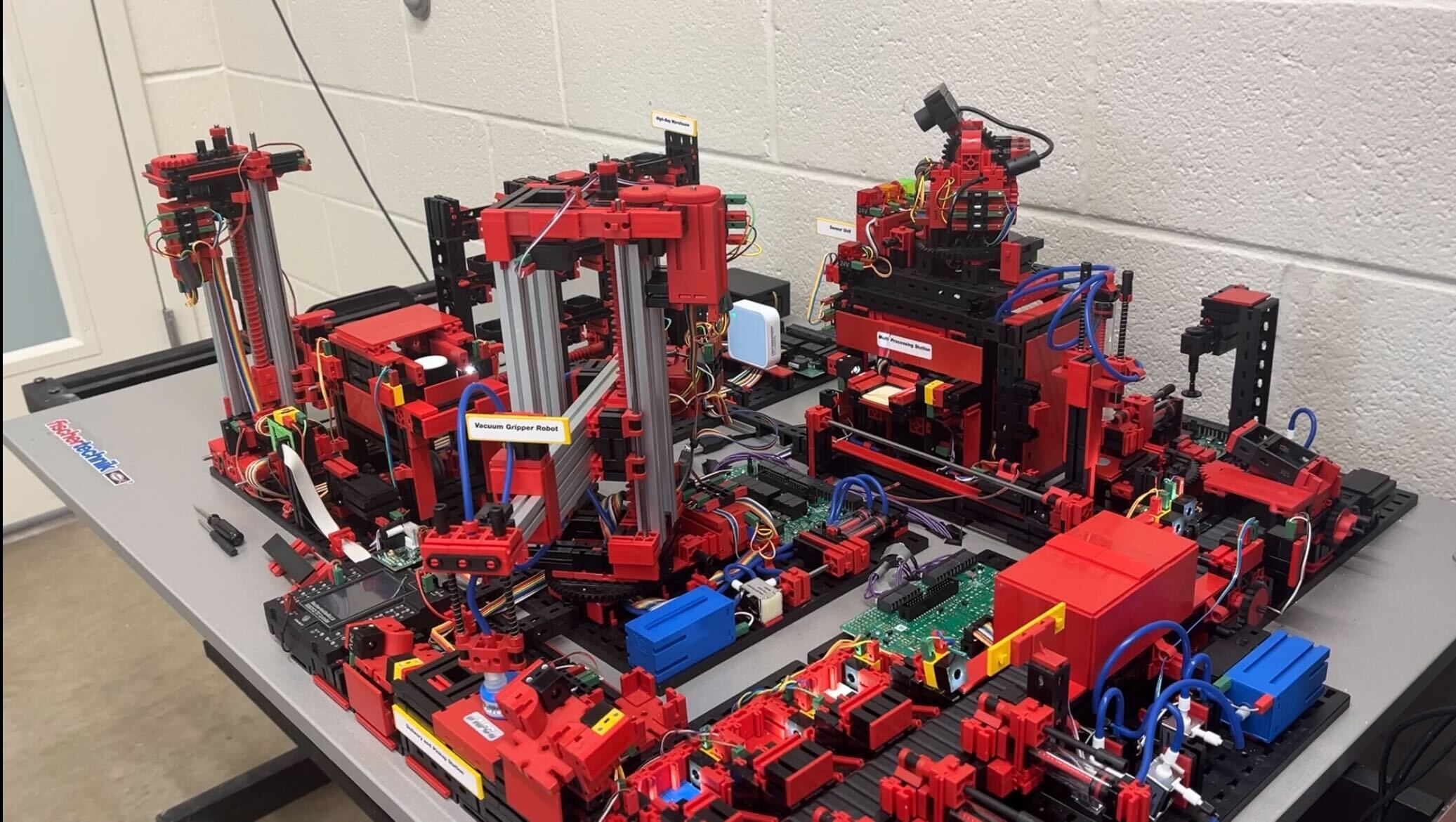}
        \vphantom{\includegraphics[width=.5\linewidth, page=14, valign=m]{images/cdc_recovery_models.pdf}}
    }
    \quad
    \subfloat[Plant model of manufacturing testbed]{
        \label{fig:case-study-mfg-plant}
        \centering
        \includegraphics[width=.5\linewidth, page=14, valign=m]{images/cdc_recovery_models.pdf}
    }
    \par 
    \caption{Fischertechnik manufacturing testbed}
    \label{fig:case-study}
\end{figure*}

Figure~\ref{fig:case-study-mfg-plant} depicts the model of this system, whereas Table~\ref{tab-case study} explains each event.
For simplicity, we model the gripper picking up a part and storing it with events $m$ (manual sort) and $s$ (automatic sort).
For instance, if the gripper picks a part in the initial location of the conveyor, the part will be manually sorted, event $m$.
When the gripper picks a red part in the buffer of location 1, the part is automatically stored sorted, event $s$.
If the gripper picks a blue part in the buffer of location 1, the part must be manually sorted before storing.
\begin{table}
\hspace{.2cm}
\caption{Event description manufacturing model}
\centering 
    \begin{tabular}[b]{l l}
         \multirow{2}{*}{Event}&
         \multirow{2}{*}{Description}\\
         &\\
         \toprule
         &\\[-1.5ex]
         $in_{I}$ & Conveyor moving the part to initial position\\
         $in_{1}$ & Conveyor moving the part to position $1$\\
         $in_{2}$ & Conveyor moving the part to position $2$\\
         $p_{1}$ & Pusher $1$ pushing and retracting\\
         $p_{2}$ &  Pusher $2$ pushing and retracting\\
         $s$ & Part sorted and stored\\
         $m$ & Part stored for manual sorting\\
         \bottomrule
         
\end{tabular}
\label{tab-case study}
\vspace{-10pt}
\end{table}

A supervisor is designed with the specification to automatically store the parts sorted by color.
The supervisor directs red parts to location 1, while blue parts go to location 2.
In Fig.~\ref{fig:case-study-mfg-plant}, states in red are unsafe states.
They represent that a part was wrongly stored as sorted, e.g., a red part picked in location 2 and stored as sorted.
States in yellow are also removed by the supervisor since they are undesired states, e.g., manually sorting parts.

\subsection{Actuator attacks}
We assume three actuator attack scenarios for the pneumatic cylinders: cylinder $p_1$ under attack $\Sigma_{c,v}^1 = \{p_1\}$, cylinder $p_2$ under attack  $\Sigma_{c,v}^2 = \{p_2\}$, and cylinders $p_1$ and $p_2$ under attacks  $\Sigma_{c,v}^3 = \{p_1,p_2\}$.
For each of these three scenarios, we construct the attacked plants $G^a_1$, $G^a_2$, and $G^a_3$ and the attacked supervisors $S^a_1$, $S^a_2$, and $S^a_3$.
We verify that in all cases, the controlled systems are AE-safe controllable.
Intuitively, the supervisor can disable event $s$ when a part goes to an incorrect position, e.g., when the attacker controller event $p_1$ pushes a blue part in the buffer of location 1. 
Next, we investigate robust-recovery strategies for these three scenarios.

\subsection{Robust-recovery strategies}
To analyze robust-recovery strategies in our case study, we first define the robust region $G_{robust}$. 
$G_{robust}$ is defined by states $A, R, RI, B$, and $BI$.
Upon attack detection, the controlled system should transition to manual sorting to maintain operational continuity.
Specifically, when a part gets in the initial position, only event $m$ should be allowed. 
By applying such a restriction within the system, further vulnerabilities are avoided.

% Consequently, $G_{robust}$ marks the following language: $(r\;in_I\;m\; + b\;in_I\;m)^*$.}

In the first scenario ($\Sigma_{c,v}^1 = \{p_1\}$), attacks are only feasible when blue parts are being processed. Specifically, after a blue part reaches position 1 (event $in_1$), an attacker can enable event $p_1$ when it should be disabled. Following this attack, the system reaches the detection state $(BM, A)$, from which $G_{robust}$ can be reached through event $m$, allowing the part to be manually sorted.

If the attack on $p_1$ occurs after a part reaches position 2 (event $in_2$), the system reaches the detection state $(BC, A)$. 
From there, $G_{robust}$ can be reached through the sequence $p_2s$, enabling recovery.
% Consider the plant model for the manufacturing system shown in Figure~\ref{fig:case-study-mfg-plant}, and the actuator attack scenario where $\Sigma_{c,v}^1 = \{p_1\}$. 
The resilient supervisor $S^{R,1}$ is shown in Figure~\ref{fig:sr-scenario1}. 

\begin{figure}
    \centering
    \includegraphics[width=1\linewidth, page=15]{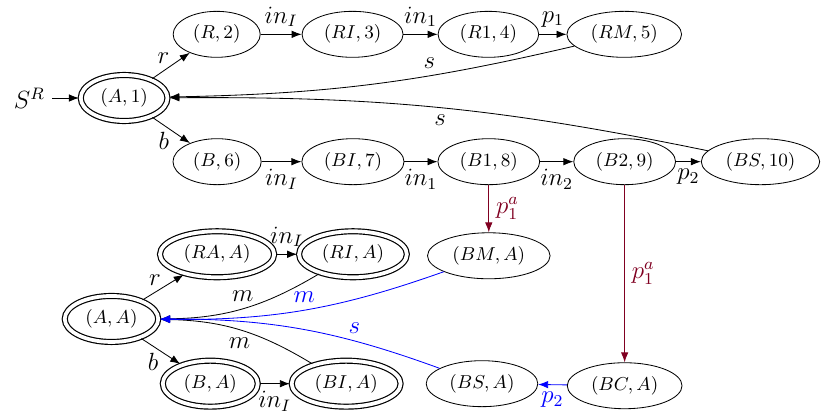}
    \caption{Resilient supervisor $S^R$ for the attack scenario where $\Sigma_{c,v}^1 = \{p_1\}.$}
    \label{fig:sr-scenario1}
    \vspace*{-2em}
\end{figure}

In the second scenario ($\Sigma_{c,v}^2 = \{p_2\}$), a similar recovery strategy exists. For instance, when $p_2$ is attacked after a red part reaches position 1, recovery is achieved through the sequence $p_1s$. When $p_2$ is attacked after a blue part reaches position 1, the system recovers through the sequence $p_1m$.

The third scenario ($\Sigma_{c,v}^3 = \{p_1,p_2\}$) presents a critical vulnerability. 
When both cylinders are vulnerable to attacks, the system $G^{R,3}$ is not AE-robust recoverable. 
For instance, when a red part reaches position 1 and event $p_2$ is attacked, the only potential recovery path to $G_{robust}$ requires using event $p_1$, which is also vulnerable to attacks. 
With both recovery paths compromised, $S^{R,3}$ cannot enforce a robust-recovery strategy upon attack detection. 
% Due to space limitation, we omit the models of supervisors $S^{R,2}$ and $S^{R,3}$.}

% \rmg{
% I will leave it to be written later: I believe that scenarios 1 and 2 are robust recoverable, while scenario 3 is not.
% }

% \begin{example}
    
% \end{example}

% Each station corresponds to a state, and completing the task at any given state triggers an event that moves a part from one station to another. 
% The goal here is to place the part into storage drop-off after the completion of heat treatment, machining and, sorting. 
% FigureXXX, illustrates the automaton of plant $G$, the system. 
% TableXXX summarizes the state names for each machine. 
% Here, we assume that the sorting station is subject to a cyber-attack which causes the pneumatic cylinders in the sorting station to randomly actuate. 
% Further, we assume that we can manually sort the part at the storage units, but this not preferred. 
% Here, the light sensor at the onset on the pneumatic cylinders represents the detection state.

% FigureXXX, depicts the resilient supervisor $S_{r}$ before the attack. 
% Once the cyber-attack is identified, the supervisor disables the event that lead to stateXXX, eventXXX. 
% The string $s=XXX$ represents the recovery path after the attack.

\section{CONCLUSIONS}\label{sec:conclusions}

In this paper, we have considered CPDES under AE-attacks. As a mitigation strategy, we formalized robust-recovery strategies that drive attacked systems to a robust region which, although more constrained than the nominal behavior, remains resilient to such attacks. We formalized the property of AE-robust recoverability and proposed a synthesis method based on the standard supervisory control paradigm, enabling the creation of supervisors that enforce robust-recovery in CPDES under AE-attacks. Relaxing the conditions for robust-recovery represents a potential extension of our work. Additionally, we plan to investigate different attack models, with particular attention to sensor attacks.
% \lipsum[2-5]

%%%%%%%%%%%%%%%%%%%%%%%%%%%%%%%%%%%%%%%%%%%%%%%%%%%%%%%%%%%%%%%%%%%%%%%%%%%%%%%%

%%%%%%%%%%%%%%%%%%%%%%%%%%%%%%%%%%%%%%%%%%%%%%%%%%%%%%%%%%%%%%%%%%%%%%%%%%%%%%%%


% \cite{Lafortune:2021}

\bibliographystyle{IEEEtran}
\bibliography{root} %ADD your bib here by adding a comma after my and adding a filename.bib to the project

\end{document}